\documentclass[prodmode,acmtkdd]{acmsmall} 

\usepackage[ruled, noend, noline, linesnumbered]{algorithm2e}
\usepackage{amsmath,amsfonts,amssymb}

\usepackage{xspace, paralist, graphicx, epstopdf, boxedminipage, subfig}
\usepackage[colorlinks,urlcolor=blue,citecolor=blue,linkcolor=blue]{hyperref}
\usepackage{subfig}
\usepackage{tabularx}
\usepackage{array}
\usepackage{subfig}
\usepackage{tabularx}
\usepackage{array}
\usepackage{booktabs}
\usepackage{textcase}

\usepackage{tikz}
\usetikzlibrary{arrows,positioning,backgrounds}
\def\typea{blue}
\def\typeb{red}

\def\doctype{2}
\def\tsubmission{1}

\ifnum\doctype=\tsubmission
\newcommand{\full}[1]{}
\newcommand{\submit}[1]{#1}
\else
\newcommand{\full}[1]{#1}
\newcommand{\submit}[1]{}
\fi

\newtheorem{claim}{Claim}

\newcommand{\set}[1]{\{#1\}}
\newcommand{\E}{\mathbf{E}}
\newcommand{\EX}{\mathbf{E}}

\newcommand{\Sec}[1]{\hyperref[sec:#1]{\S\ref*{sec:#1}}} 
\newcommand{\Eqn}[1]{\hyperref[eqn:#1]{(\ref*{eqn:#1})}} 
\newcommand{\Fig}[1]{\hyperref[fig:#1]{Fig.\,\ref*{fig:#1}}} 
\newcommand{\Tab}[1]{\hyperref[tab:#1]{Tab.\,\ref*{tab:#1}}} 
\newcommand{\Thm}[1]{\hyperref[thm:#1]{Thm.\,\ref*{thm:#1}}} 
\newcommand{\Lem}[1]{\hyperref[lem:#1]{Lem.\,\ref*{lem:#1}}} 
\newcommand{\Prop}[1]{\hyperref[prop:#1]{Prop.~\ref*{prop:#1}}} 
\newcommand{\Cor}[1]{\hyperref[cor:#1]{Cor.~\ref*{cor:#1}}} 
\newcommand{\Def}[1]{\hyperref[def:#1]{Defn.~\ref*{def:#1}}} 
\newcommand{\Alg}[1]{\hyperref[alg:#1]{Alg.\,\ref*{alg:#1}}} 
\newcommand{\Ex}[1]{\hyperref[ex:#1]{Ex.~\ref*{ex:#1}}} 
\newcommand{\Clm}[1]{\hyperref[clm:#1]{Claim~\ref*{clm:#1}}} 
\newcommand{\Step}[1]{\hyperref[step:#1]{Step~\ref*{step:#1}}} 

\newcommand{\cC}{\mathcal{C}}
\newcommand{\cD}{\mathcal{D}}
\newcommand{\cF}{\mathcal{F}}
\newcommand{\cN}{\mathcal{N}}
\newcommand{\cR}{\mathcal{R}}
\newcommand{\cS}{\mathcal{S}}
\newcommand{\cW}{\mathcal{W}}
\newcommand{\cE}{{\cal E}}

\newcommand{\algstream}{{\sc Streaming-Triangles}}
\newcommand{\algupdate}{{\sc Update}}
\newcommand{\algbit}{{\sc Single-Bit}}

\newcommand{\isClosed}{{\it isClosed}}
\newcommand{\true}{{\tt true}}
\newcommand{\false}{{\tt false}}
\newcommand{\wedgeres}{{\it wedge\_res}}
\newcommand{\edgeres}{{\it edge\_res}}
\newcommand{\newwedges}{{\it new\_wedges}}
\newcommand{\totwedges}{{\it tot\_wedges}}

\newcommand{\stor}{s}
\newcommand{\singlestor}{|\cR|}
\newcommand{\gcc}{\kappa}

\newcommand{\ignore}[1]{}

\begin{document}

\markboth{M. Jha et al.}{A space efficient streaming algorithm for  estimating transitivity and triangle counts}

\title{A space efficient streaming algorithm for  estimating transitivity and triangle counts using the birthday paradox}
\author{MADHAV JHA
\affil{Sandia National Laboratories}
C. SESHADHRI
\affil{Sandia National Laboratories}
ALI PINAR
\affil{Sandia National Laboratories}}

\begin{abstract}
We design a space efficient algorithm that approximates the transitivity (global clustering coefficient)
and total triangle count
with only a single pass through a graph given as a stream of edges. Our procedure is based on the classic probabilistic
result, \emph{the birthday paradox}. 
When the transitivity is constant and there are more edges than wedges (common properties for social networks),
we can prove that our algorithm requires $O(\sqrt{n})$ space ($n$ is the number of vertices) to provide accurate estimates.
We run a detailed set of experiments on a variety of real graphs and demonstrate
that the memory requirement of the algorithm is a tiny fraction of the graph. For example, even for a graph
with 200 million edges, our algorithm stores just 60,000 edges to give accurate results.
Being a single pass streaming algorithm, our procedure also
maintains a real-time estimate of the transitivity/number of triangles of a graph, by storing
a minuscule fraction of edges. \\
\end{abstract}




\begin{bottomstuff}This manuscript is an extended version of \cite{JhSePi13}.\\
This work was funded by the GRAPHS program under DARPA,  Complex Interconnected Distributed Systems
    (CIDS) program  DOE Applied Mathematics Research program  and under Sandia's Laboratory Directed Research \& Development (LDRD) program. Sandia National Laboratories is a multi-program
    laboratory managed and operated by Sandia Corporation, a wholly   owned subsidiary of Lockheed Martin Corporation, for the   U.S. Department of Energy's National Nuclear Security
    Administration under contract DE-AC04-94AL85000.
\end{bottomstuff}

\maketitle
 
\section{Introduction}
Triangles are one of the most important motifs in real world networks. Whether the networks
come from social interaction, computer communications, financial transactions, proteins, or  ecology,  the abundance of triangles
is pervasive, and this abundance is a critical feature  that distinguishes real graphs from random graphs.  
There is a rich body of literature on analysis of triangles and counting algorithms.
Social scientists use triangle counts to understand graphs~\cite{Co88,Po98,Burt04,FoDeCo10};
graph mining applications such as spam detection and finding common topics on the {\sc WWW} use triangle counts~\cite{EcMo02,BeBoCaGi08};
motif detection in bioinformatics often count the frequency of triadic patterns~\cite{Milo2002}.
Distribution of degree-wise clustering coefficients was used as the driving force for a new generative model, Blocked Two-Level Erd\"os-R\'enyi~\cite{SeKoPi11}.  
Durak et al. observed that  the relationships among degrees of triangle vertices can be a descriptor of the underlying graph~\cite{DuPiKo12}. 
Nevertheless, counting triangles continues  to be a challenge due to sheer sizes of the graphs (easily in the order of billions of edges).

\begin{figure*}[tb]
  \centering
  \subfloat[Transitivity]{\includegraphics[width=0.5\textwidth]{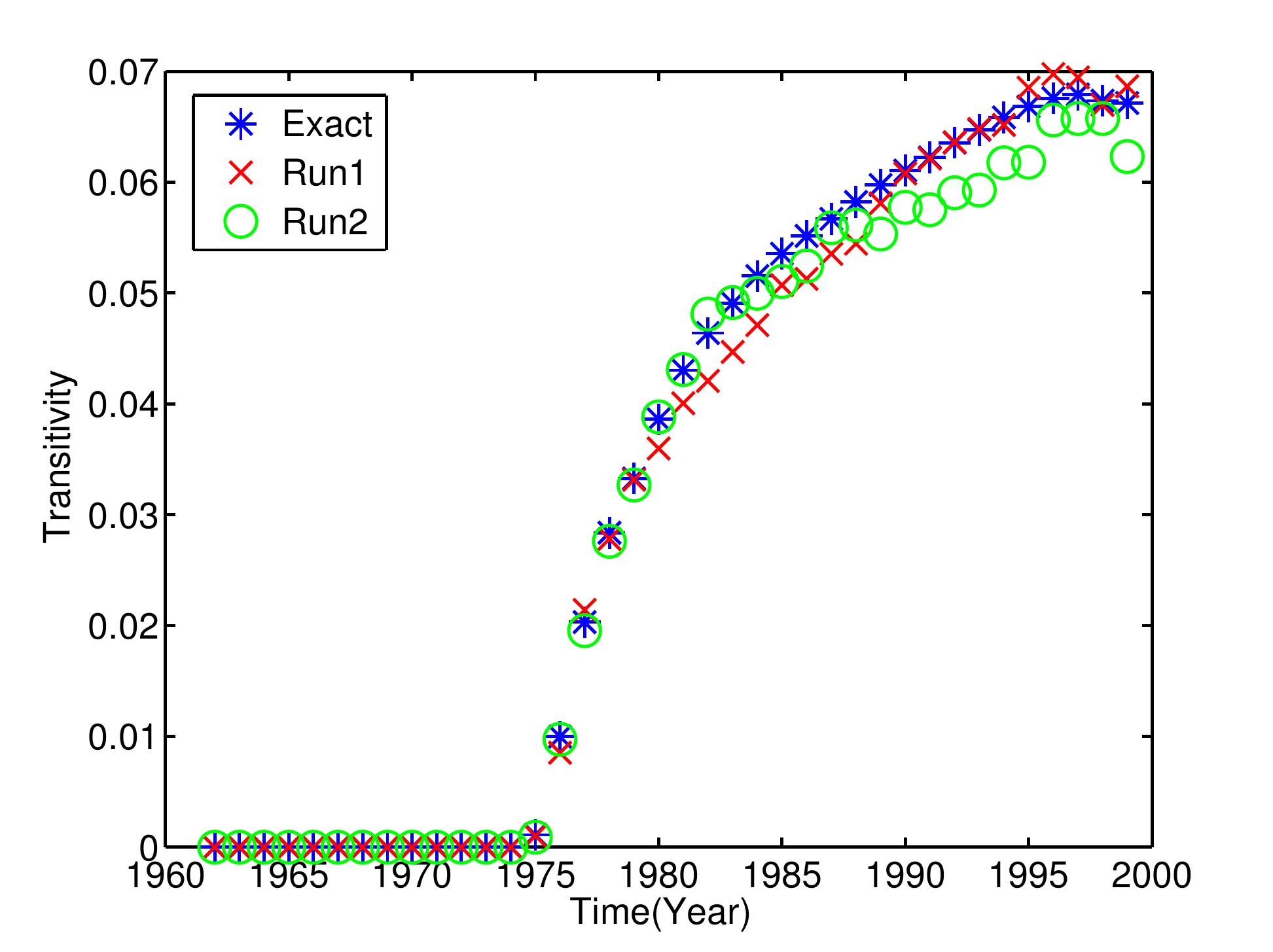}}
  \subfloat[Triangle count]{\includegraphics[width=0.5\textwidth]{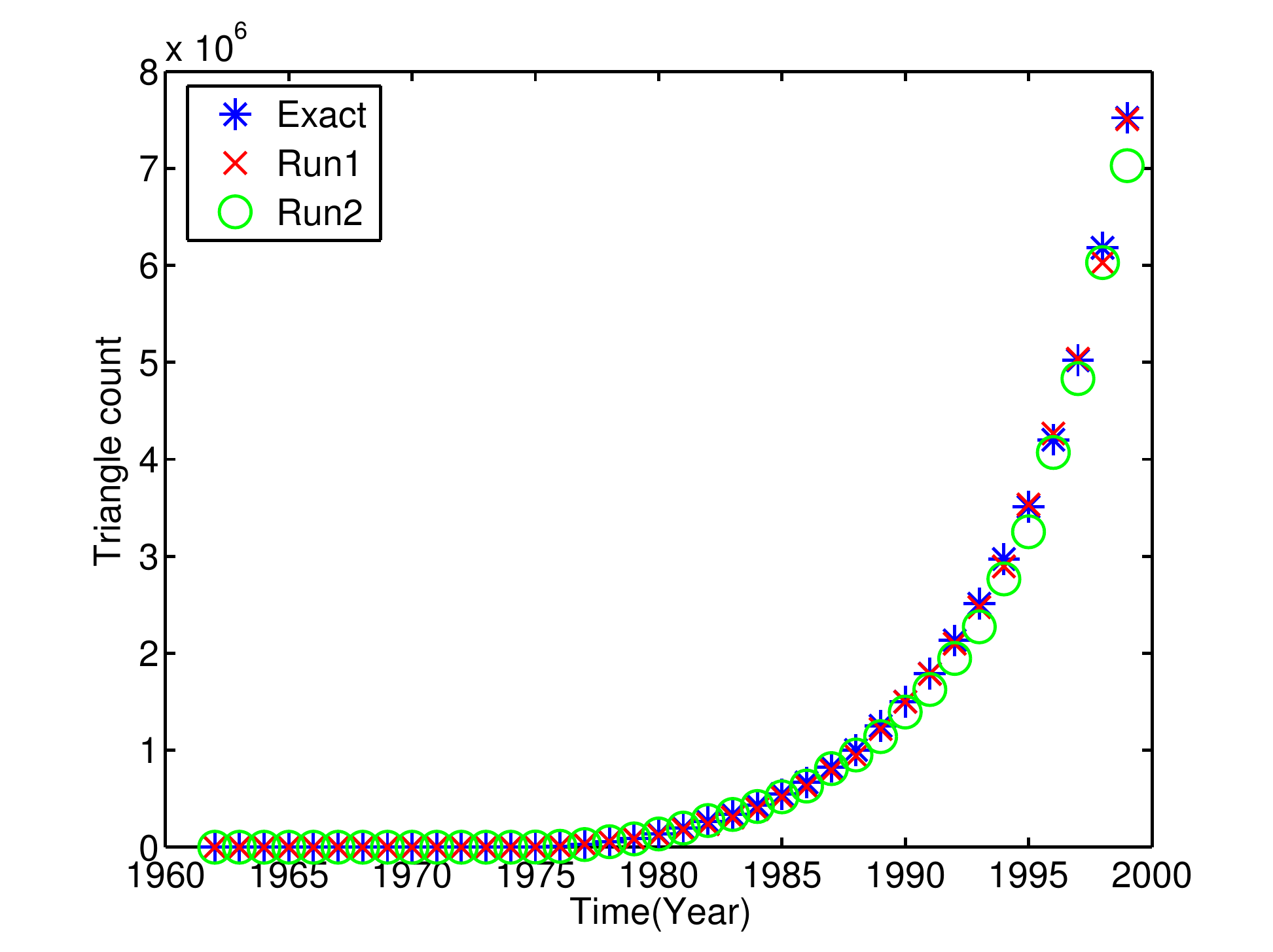}}
   \caption{Realtime tracking of number of triangles and transitivities on cit-Patents (16M edges), storing only 100K edges from the past.}
  \label{fig:tracking}
\end{figure*}

Many massive graphs come from modeling interactions in a dynamic system. 
People call each other on the phone, exchange emails, or co-author a paper;  computers exchange messages;  animals come in the vicinity of each other;  companies trade with each other. 
These interactions manifest as a \emph{stream of edges}. The edges appear with timestamps, or ``one at a time."  The network (graph) that represents the system  is an accumulation of the observed edges.  
There are many methods to deal with such massive graphs, such as random sampling~\cite{ScWa05-2,TsKaMiFa09,SePiKo13}, MapReduce paradigm~\cite{SuVa11,Pl12}, distributed-memory parallelism~\cite{ArKhMa12,Ch11}, adopting external memory~\cite{Ch95,ArGoSi10}, and multithreaded parallelism~\cite{BeHeKaKo07}. 

All of these methods however, need to store at least a large fraction of the data. On the other hand, 
a \emph{small space streaming algorithm} maintains a very small (using randomness)
set of edges, called the ``sketch", and updates this sample as edges appear. 
Based on the sketch and some auxiliary data structures, the algorithm computes an accurate estimate for the number of triangles for the graph seen so far.
The sketch size is orders of magnitude smaller than the total graph. Furthermore,
it can be updated rapidly when new edges arrive and hence maintains a real-time estimate of the number of triangles.
We also want a single pass algorithm, so it only observes each edge once (think of it as making a single scan
of a log file). The algorithm cannot revisit edges that it has forgotten.

\subsection{The streaming setting} \label{sec:problem}
Let $G$ be a simple undirected graph with $n$ vertices and $m$ edges.
Let $T$ denote the  number of triangles in the graph and $W$ be the number of \emph{wedges}, where a {\em wedge} is a path of length $2$.
A common measure is the \emph{transitivity} $\gcc = 3T/W$~\cite{WaFa94}, a 
measure of how often friends of friends are also friends. (This is also called the global
clustering coefficient.)

A single pass streaming algorithm is defined as follows. Consider a sequence
of distinct edges $e_1, e_2, \ldots, e_m$. Let $G_t$ be the \emph{graph at time $t$},
formed by the edge set $\{e_i | i \leq t\}$. The stream of edges can be considered as a sequence of 
\emph{edge insertions} into the graph. Vertex insertions can be  handled trivially.
We do not know the number of vertices ahead of time and simply see each edge as a pair
$(u,v)$ of vertex labels. New vertices are implicitly added as new labels.
There is no assumption on the order of edges in the stream. Edges incident to a single vertex
do not necessarily appear together.

In this paper, \emph{we do not consider edge/vertex deletions or repeated edges.} In that sense, this is a simplified
version of the full-blown streaming model.
Nonetheless, the edge insertion model on simple graphs is the standard for previous
work on counting triangles~\cite{BaKuSi02,JoGh05,BuFrLeMaSo06,AhGuMc12,KaMeSaSu12}.

A streaming algorithm has a small memory, $M$, and sees the edges in stream order. At each edge, $e_t$, the algorithm can choose
to update data structures in $M$ (using the edge $e_t$). 
Then the algorithm proceeds to $e_{t+1}$, and so on. The algorithm
is never allowed to see an edge that has already passed by.
The memory $M$ is much smaller than $m$, so the algorithm keeps
a small ``sketch" of the edges it has seen.
The aim is to estimate the number of triangles in $G_m$ at the end of the stream.
Usually, we desire the more stringent guarantee of maintaining a running estimate
of the number of triangles and transitivity of $G_t$ at time $t$. 
We denote these quantities respectively as $T_t$ and $\gcc_t$.	

\subsection{Results} \label{sec:results}

We present a single pass, $O(m/\sqrt{T})$-space algorithm to provably estimate the transitivity (with arbitrary
additive error) in a streaming graph. Streaming algorithms for counting triangles or computing the transitivity have been studied before,
but no previous algorithm attains this space guarantee.
Buriol et al.~\cite{BuFrLeMaSo06} give a single pass algorithm with a stronger relative error guarantee 
that requires space $O(mn/T)$. We discuss in more detail later.

Although our theoretical result is interesting asymptotically, the constant factors and dependence
on error in our bound are large. Our main result is a \emph{practical} streaming 
algorithm (based on the theoretical one) for computing $\gcc$ and $T$, using additional probabilistic heuristics.
We perform an extensive empirical analysis of our algorithm on a variety
of datasets from SNAP~\cite{Snap}.
The salient features of our algorithm are:

\begin{asparaitem}
	\item {\bf Theoretical basis:} Our algorithm is based on the classic \emph{birthday paradox}:
	if we choose $23$ random people, the probability that $2$ of them share a birthday is at least $1/2$ (Chap. II.3 of \cite{Fel50}).
	We extend this analysis for sampling wedges in a large pool of edges. The final streaming algorithm
	is designed by using reservoir sampling with wedge sampling \cite{SePiKo13} for estimating $\gcc$.
	We prove a space bound of $O(m/\sqrt{T})$, which we show is $O(\sqrt{n})$ under common conditions
	for social networks. In general, the number of triangles, $T$ is fairly large for many real-world graphs, and this
	is what gives the space advantage. 

	While our theory appears to be a good guide in designing the algorithm and explaining its behavior, 
it should not be used to actually decide space bounds in practice. For graphs where
$T$ is small, our algorithm does not provide good guarantees with small space (since $m/\sqrt{T}$ is large).

	\item {\bf Accuracy and scalability with small sketches:} We test our algorithm on a variety of graphs from
	different sources. In all instances, we get accurate estimates for $\gcc$ and $T$ by storing
	at most 40K edges. This is even for graphs where $m$ is in the order of millions. 
	Our relative errors on $\gcc$ and the number of triangles are mostly less than 5\% (In a graph with
	very few triangles where $\gcc < 0.01$, our triangle count estimate has relative error of 12\%). 
	Our algorithm can process extremely large graphs. Our experiments include a run on a streamed Orkut social network with 200M edges
	(by storing only 40K edges, relative errors are at most 5\%). We get similar results on streamed Flickr and Live-journal graphs with tens
	of millions of edges. 
	
	We run detailed experiments on some test graphs (with 1-3 million edges) with varying parameters to show convergence of our algorithm.
	Comparisons with previous work~\cite{BuFrLeMaSo06} show that our algorithm gets within 5\% of the true answer,
	while the previous algorithm is off by more than 50\%.
		 
	\item {\bf Real-time tracking:} For a temporal graph, our algorithm precisely tracks both $\gcc_t$
	and $T_t$ with less storage. By storing 60K edges of the past, we can track this
	information for a patent citation network with 16 million edges \cite{Snap}. Refer to \Fig{tracking}. We maintain
	a real-time estimate of both the transitivity and number of triangles with a single
	pass, storing less than 1\% of the graph. We see some fluctuations in the transitivity estimate due
	to the randomness of the algorithm, but the overall tracking is consistently accurate. 
\end{asparaitem}

\subsection{Previous work} \label{sec:prev}
Enumeration of all triangles is a well-studied problem~\cite{ChNi85,ScWa05,latapy08,BeFoNoPh11,ChCh11}.
Recent work by Cohen~\cite{Co09}, Suri and Vassilvitskii~\cite{SuVa11}, Arifuzzaman et al.~\cite{ArKhMa12} 
 give  massively parallel implementations
of these algorithms. 
Eigenvalue/trace based methods have also been used ~\cite{Ts08,Av10} to compute
estimates of the total and per-degree number of triangles. 

Tsourakakis et al.~\cite{TsDrMi09} started the use of sparsification methods, the most important of which
is Doulion~\cite{TsKaMiFa09}.  
Various analyses of this algorithm (and its variants) have been
proposed~\cite{KoMiPeTs10,TsKoMi11,YoKi11,PaTs12}. 
Algorithms based on wedge-sampling provide provable accurate estimations on various triadic measures on graphs~\cite{ScWa05-2,SePiKo13}.
 Wedge sampling techniques have also been applied to directed graphs~\cite{SePiKo13-2} and implemented with MapReduce~\cite{KoPiPlSe13}.
 
Theoretical streaming algorithms for counting triangles were initiated by Bar-Yossef et al.~\cite{BaKuSi02}.
Subsequent improvements were given in \cite{JoGh05,BuFrLeMaSo06,AhGuMc12,KaMeSaSu12}.
The space bounds achieved are of the form $mn/T$. Note that $m/\sqrt{T} \leq mn/T$ 
whenever $T \leq n^2$ (which is a reasonable assumption for sparse graphs). 
These algorithms are rarely
practical, since $T$ is often much smaller than $mn$. Some multi-pass streaming algorithms
give stronger guarantees, but we will not discuss them here.

Buriol et al.~\cite{BuFrLeMaSo06} give an implementation 
of their algorithm. For almost all of their experiments on graphs, with storage of 100K edges,
they get fairly large errors (always more than 10\%, and often more than 50\%). 
Buriol et al. provide an implementation in the incidence list setting, where all neighbors of a vertex arrive together.
In this case, their algorithm is quite practical since the errors are quite small.
Our algorithm scales to sizes (100 million edges) larger than their experiments. We get better
accuracy with far less storage, without any assumption on the ordering of the data stream. Furthermore, our algorithm
performs accurate real-time tracking.

Becchetti et al.~\cite{BeBoCaGi08} gave a semi-streaming algorithm for counting the triangles
incident to \emph{every} vertex. Their algorithm uses clever methods to approximate Jaccard similarities,
and requires multiple passes over the data.  Ahmed et al. studied sampling a subgraph from a stream of edges that preserves multiple properties of the original graph~\cite{AhJeKo13}.
Our earlier results  on triadic measures were presented in~\cite{JhSePi13}. 
More recently, Pavan et al.~\cite{PavanVLDB2013} introduce an approach called {\em neighborhood sampling} for estimating triangle counts which gives a 1-pass streaming algorithm with space bound  $O(m\Delta/T)$,  where $\Delta$ is the maximum degree of the graph.   Their implementation is practical and achieves good accuracy estimates on the lines
of our practical implementation. \cite{TangwongsanCIKM2013} explores a parallel implementation of \cite{PavanVLDB2013}. (As a minor comment, our algorithm gets good results by storing less than 80K edges, while~\cite{PavanVLDB2013} only shows comparable results for storing 128K ``estimators", each of which at least stores an edge.)

\subsection{Outline} \label{sec:outline}
A high-level description of our  practical algorithm  \algstream{} is presented in \Sec{highlevel}. 
We start with  the intuition behind the algorithm, followed by a detailed description of the implementation. 
\Sec{algorithm} provides a  theoretical analysis for an
idealized variant called \algbit{}.
We stress that \algbit{} is a thought experiment
to highlight the theoretical aspects of our result, and we do not actually implement it.  Nevertheless,  this algorithm forms that basis of  a practical algorithm, and in \Sec{circum},
we explain the heuristics used to get \algstream. \Sec{analysis} gives
an in-depth mathematical analysis of \algbit.

In \Sec{experiments}, we give various empirical results of our runs of \algstream{} on real graphs.
We show that na\"{i}ve implementations based on \algbit{} perform poorly in practice, and we need
our heuristics to get a practical algorithm. 

\section{The Main Algorithm}
\label{sec:highlevel}
\subsection{Intuition for the algorithm} \label{sec:int}

The starting point for our algorithm is the idea of \emph{wedge sampling} to estimate the transitivity, $\gcc$~\cite{SePiKo13}. A wedge is \emph{closed} if it participates in a triangle and \emph{open} otherwise.
Note that $\gcc = 3T/W$ is exactly the probability that a uniform random wedge
is closed. This leads to a simple randomized algorithm for estimating $\gcc$ (and $T$), by generating
a set of (independent) uniform random wedges and finding the fraction that are closed.
But how do we sample wedges from a stream of edges?

Suppose we just sampled a uniform random
set of edges. How large does this set need to be to get a wedge? The \emph{birthday paradox}
can be used to deduce that (as long as $W \geq m$, which holds for a great majority, if not all, of real  networks) $O(\sqrt{n})$ edges suffice. 
 A more sophisticated 
result, given in \Lem{wedge-via-collisions}, provides (weak) concentration bounds on
the number of wedges generated by a random set of edges. A
``small" number of uniform random edges can give enough wedges to perform wedge sampling
(which in turn is used to estimate $\gcc$). 

A set of uniform random edges can be maintained by \emph{reservoir sampling} \cite{Vi85}.
From these edges, we generate a random wedge by doing a second level of reservoir sampling. This process
implicitly treats the wedges created in the edge reservoir as a stream, and performs reservoir
sampling on that. Overall, this method approximates uniform random wedge sampling.

As we maintain our reservoir wedges, we check for closure by the future edges
in the stream. But there are closed wedges that cannot be verified, because the closing edge
may have already appeared in the past. A simple observation used by past streaming algorithms saves the day~\cite{JoGh05,BuFrLeMaSo06}.
In each triangle, there is exactly one wedge
whose closing edge appears in the future. So we try
to approximate the fraction of these ``future-closed" wedges, which is exactly one-third of the fraction of closed wedges.

Finally, to estimate $T$ from $\gcc$, we need an estimate of the total number of wedges $W$.
This can be obtained by reverse engineering the birthday paradox: given the number of wedges
in our reservoir of sample edges, we can estimate $W$ (again, using the workhorse \Lem{wedge-via-collisions}).

\subsection{The procedure \large{\algstream}} \label{sec:algstream}
The streaming algorithm maintains two primary data structures: the \emph{edge reservoir} and the \emph{wedge reservoir}. The edge reservoir maintains a uniform random sample of edges  observed so far. The wedge reservoir aims to select a uniform sample of  wedges.
Specifically, it maintains a uniform sample of the wedges created by the edge reservoir at any step of the process.  
(The wedge reservoir may include wedges whose edges are no longer  in the  edge reservoir.)
The two parameters for the streaming algorithm are $\stor_e$ and $\stor_w$, the sizes of  edge and wedge pools, respectively.
The main algorithm is described in \algstream{}, although most of the technical
computation is performed in \algupdate{}, which is invoked every time a new edge appears.

After edge $e_t$ is processed by \algupdate, the algorithm computes running estimates for $\gcc_t$ and $T_t$. 
These values do not have to be stored, so they are immediately output.
We describe the main data structures of the algorithm \algstream.

\begin{compactitem}
	\item Array \edgeres$[1\cdots \stor_e]$: This is the array of \emph{reservoir edges} and is the subsample of the stream maintained. 
	\item New wedges $\cN_t$: This is a list of all wedges involving $e_t$ formed only by edges in \edgeres. This
	may often be empty, if $e_t$ is not added to the \edgeres. We do not necessarily maintain this list explicitly, and we discuss implementation details later.
	\item Variable \totwedges: This is the total number of wedges formed by edges in the current \edgeres.
	\item Array \wedgeres$[1 \cdots s_w]$: This is an array of \emph{reservoir wedges} of size $s_w$. 
	\item Array \isClosed$[1 \cdots s_w]$: This is a boolean array. We set \isClosed$[i]$ to be \true{}
	if wedge \wedgeres$[i]$ is detected as closed.
\end{compactitem}
\medskip
On seeing edge $e_t$, \algstream{} updates the data structures.
The estimates $\gcc_t$ and $T_t$ are computed using the fraction of \true{} bits
in \isClosed, and the variable \totwedges.

\begin{algorithm}
 \caption{\algstream($\stor_e,s_w$)}\label{alg:stream}
 \DontPrintSemicolon
Initialize \edgeres{} of size $\stor_e$ and \wedgeres{} of size $s_w$.
For each edge $e_t$ in stream,\;
\ \ \ \ Call \algupdate($e_t$).\;
\ \ \ \ Let $\rho$ be the fraction of entries in \isClosed{} set to \true.\;
\ \ \ \ Set $\gcc_t = 3\rho$.\;
\ \ \ \ Set $T_t = [\rho t^2/s_e(s_e-1)] \times \totwedges$.\;
\end{algorithm}

\algupdate{} is where all the work happens, since it processes each edge $e_t$ as it arrives.
Steps~\ref{step:one}--\ref{step:3} determine all the wedges in  the wedge reservoir that are closed by $e_t$ and updates \isClosed{} accordingly.
In Steps~\ref{step:res1}-\ref{step:res4}, we perform reservoir sampling on \edgeres{}.
This involves replacing each entry by $e_t$ with probability $1/t$. The  remaining steps are executed
iff this leads to any changes in \edgeres. We perform some updates to \totwedges{} and determine the
new wedges $\cN_t$. Finally, in Steps~\ref{step:wed}-\ref{step:wed-fin}, we perform reservoir sampling on \wedgeres{}, where each entry
is randomly replaced with some wedge in $\cN_t$. Note that we may remove wedges that have already closed.

\begin{algorithm}
 \caption{\algupdate($e_t$)}\label{alg:update}
 \DontPrintSemicolon
{\bf for} $i=1,\ldots, s_w$ \\ \label{step:one}
\ \ {\bf if}  \wedgeres{}$[i]$ closed by $e_t$ \;
\ \ \ \ \isClosed$[i] \leftarrow $ \true \; \label{step:3}
{\bf for} $i=1,\ldots, s_e$ \; \label{step:res1}
\ \  Pick a random number $x$ in $[0,1]$\;
\ \ {\bf if}  $ x\leq 1/t$ \;
\ \ \ \ \edgeres$[i]  \leftarrow e_t$.\label{step:res4}\;
{\bf if}  there were any updates of \edgeres \\
\ \  Update \totwedges{}, the number of wedges formed by \edgeres.\;
\ \ Determine $\cN_t$ (wedges involving $e_t$) and let \newwedges{} $= |\cN_t|$.\;
\ \ {\bf for }  $i,\ldots s_w$ ,\label{step:wed}\;
\ \ \ \  Pick a random number $x$ in $[0,1]$\;
\ \ \ \ {\bf if}  $ x\leq  \newwedges/\totwedges$ \;
\ \ \ \ \ \ Pick uniform random  $w \in \cN_t$.\;
\ \ \ \ \ \  \wedgeres$[i] \leftarrow w$.\\
\ \ \ \ \ \  \isClosed$[i] \leftarrow $ \false.\label{step:wed-fin}\;
\end{algorithm}

\subsection{Implementation details} \label{sec:implement}

Computing
$\gcc_t$ and $T_t$ are simple and require no overhead. We maintain \edgeres{}
as a time-variable subgraph.
Each time \edgeres{} is updated, the subgraph
undergoes an edge insert and edge delete. Suppose $e_t = (u,v)$.
Wedges in $\cN_t$ are given by the neighbors of $u$ and $v$ in this subgraph.
From random access to the neighbor lists of $u$ and $v$, we can generate
a random wedge from $\cN_t$ efficiently.

Updates to the edge reservoir  are very infrequent.
At time $t$, the probability of an update is $1 - (1-1/t)^{\stor_e}$. By linearity
of expectation, the total number of times that \edgeres{} is updated is
$$ \sum_{t \leq m} 1 - (1-1/t)^{\stor_e} \approx \sum_{t \leq m} \stor_e/t \approx \stor_e \ln m $$
For a fixed $\stor_e$, this increases very slowly with $m$. So for most steps, we neither update \edgeres{} or sample a new wedge.

The total number of edges that are stored from the past is $\stor_e + \stor_w$. The edge
reservoir explicitly stores edges, and at most $\stor_w$ edges are implicitly stored (for closure).
Regardless of the implementation, the extra data structures overhead is at most twice
the storage parameters $\stor_e$ and $\stor_w$. Since these are at least $2$ orders of magnitude
smaller than the graph, this overhead is affordable. 

\section{The idealized algorithm \large{\algbit}} 
\label{sec:algorithm}

\subsection{Description of the Algorithm}
\label{sec:ideal-desc}
The algorithm \algbit{} is an idealized variant of \algstream{} that
we can formally analyze.
It requires more memory and expensive updates, but explains
the basic principles behind our algorithm. 
We later give the memory reducing heuristics that take us from \algbit{} to \algstream{}.

The procedure \algbit{} outputs a single (random) bit, $b_t$,
at each $t$. The expectation of this bit is related to the transitivity $\gcc_t$. 
\algbit{} maintains a set of reservoir edges $\cR$ of fixed size. 
We use $\cR_t$ to denote the reservoir at time $t$; abusing notation, the size is just
denoted by $\singlestor$ since it is independent of $t$.
The set of wedges constructed from $\cR_t$ is $\cW_t$.
Formally, $\cW_t = \{ \textrm{wedge} \ (e,e') | e,e' \in \cR_t\}$. 
\algbit{} maintains a set $\cC_t$, the set of wedges in $\cW_t$ 
for which it has \emph{detected a closing edge}.
Note that this is a subset of all closed wedges in $\cW_t$.
This set is easy to update as $\cR_t$ changes.
\begin{algorithm}
 \caption{\algbit}\label{alg:bit}
 \DontPrintSemicolon
For each $e_t$ in stream,\;
\ \ \ For each edge in $\cR_{t-1}$, replace it independently by $e_t$ with probability $1/t$. This yields $\cR_t$.\;
\ \ \ Construct the set of wedges $\cW_t$.\;
\ \ \ Denoting $\cD_t$ as the set of all wedges in $\cW_t$ closed by $e_t$, update $\cC_t = (\cC_{t-1} \cap \cW_t) \cup \cD_t$.\;
\ \ \ If $\cW_t$ is empty, \;
\ \ \ \ \ \ output $b_t = 0$ \;
\ \ \ Else \;
\ \ \ \ \ \ \ Pick a uniform random  wedge in $\cW_t$.\;
\ \ \ \ \ \ \ Output $b_t = 1$ if this wedge is in $\cC_t$ and $b_t = 0$ otherwise.\;
\end{algorithm}

For convenience, we state our theorem for the final time step. However, it also holds (with an identical proof) for any
large enough time $t$. It basically argues that the expectation of $b_m$ is almost $\gcc_m/3$. 
Furthermore, $|\cW_m|$ can be used to estimate $W$.

\begin{theorem} \label{thm:algbit} Assume $W \geq m$ and fix $\beta \in (0,1)$.
Suppose $\singlestor \geq c m/(\beta^3\sqrt{T})$, for some sufficiently large constant $c$. Set $est = m^2 |\cW_m|/(\singlestor(\singlestor-1))$. 
Then $|\gcc/3 - \E[b_m]| < \beta$ and with probability $>1 - \beta$, $|W - est| < \beta W$.
\end{theorem}
 
The memory requirement of this algorithm is defined by $\singlestor$, which we assume to be $O(m/\sqrt{T})$.
We can show that $m/\sqrt{T} = O(\sqrt{n/\gcc})$ (usually much smaller for heavy tailed graphs)
when $W \geq m$. Denote the degree of vertex $v$ by $d_v$. In this case, we can bound $2W = \sum_v d_v(d_v-1) = \sum_v d^2_v - 2m \geq \sum_v d^2_v - 2W$,
so $W \geq \sum_v d^2_v/4$. By $2m = \sum_v d_v$ and the Cauchy-Schwartz inequality,
$$
\frac{m}{\sqrt{W}} \leq  \frac{\sum_v d_v}{ \sqrt{\sum_v d^2_v}} \leq  \frac{\sqrt{\sum_v 1} \sqrt{\sum_v d^2_v}}{\sqrt{\sum_v d^2_v}} = \sqrt{n}$$
Using the above bound, we get ${m}/{\sqrt T} = {\sqrt{ 3} m}/{\sqrt{\kappa W}} \leq \sqrt{3 n / \kappa}$. Hence, when $W \geq m$ and $\gcc$ is a constant (both reasonable assumptions for social networks),
we require only $O(\sqrt{n})$ space.

\subsection{Analysis of the algorithm} \label{sec:analysis}

\submit{Due to space considerations and for clarity's sake, we provide proof sketches in this version. 
Mathematical details
can be found in the online full version~\cite{JhSePi12}.}

The aim of this section is to prove \Thm{algbit}.
We begin with some preliminaries.
First, the set $\cR_t$ is a set of $\singlestor$ uniform i.i.d. samples
from $\{e_1, e_2, \ldots, e_t\}$, a direct consequence of reservoir sampling.
Next, we define \emph{future-closed} wedges. Take the final graph $G$ and label all edges with their
timestamp. For each triangle, the wedge formed by the earliest two timestamps is a \emph{future-closed wedge}. In other words,
if a triangle $T$ has edges $e_i, e_j, e_k$, ($i < j < k$), then the wedge $\{e_i, e_j\}$ is future-closed. 
The number of future-closed wedges is exactly $T$, since each triangle contains exactly one  such wedge.
We have a simple yet important claim about \algbit.

\begin{claim} \label{clm:future} The set $\cC_m$ is exactly the set of future-closed wedges in $\cW_m$.
\end{claim}

\begin{proof} Consider a wedge $\{e_i,e_j\}, i < j$ in $\cW_m$. This wedge
was formed at time $j$, and remains in all $\cW_{t}$ for $j \leq t \leq m$.
If this wedge is future-closed (say by edge $e_{t'}$, for $t' > j$), then
at time $t'$, the wedge will be detected to be closed. Since this information
is maintained by \algbit, the wedge will be in $\cC_m$. If the wedge
is not future-closed, then no closing edge will be found for it after time $j$.
Hence, it will not be in $\cC_m$.
\end{proof}

The main technical effort goes into showing that $|\cW_m|$, the number of wedges
formed by edges in $\cR_m$, can be used to determine the number
of wedges in $G_m$. Furthermore, the number of future-closed wedges in $\cR_m$ (precisely $|\cC_m|$, by \Clm{future}) 
can be used to estimate $T$.

This is formally expressed in the next lemma. Roughly, if $\singlestor = km/\sqrt{W}$,
then we expect $k^2$ wedges to be formed by $\cR_m$. We also get weak concentration bounds
for the quantity. A similar bound (with somewhat weaker concentration) holds even when we consider the set of future-closed wedges.

\begin{lemma}[Birthday paradox for wedges]\label{lem:wedge-via-collisions}
Let $G$ be a graph with $m$ edges and $\cS$ be a fixed subset of wedges in $G$.
Let $\cR$ be a set of i.i.d. uniform random edges from $G$.  
Let $X$ be the random variable denoting the number of wedges in $\cS$ formed by edges in $\cR$.

\begin{enumerate}
\item $\E[X] = {\singlestor \choose 2} (2|\cS|/m^2)$. 
\item Let $\gamma > 0$ be a parameter and $c'$ be a sufficiently large constant. Assume $W \geq m$. If $\singlestor \geq c'm/(\gamma^3 \sqrt W)$,
 then with probability at least $1 -\gamma$, $|X - \E[X]| \leq (\gamma W/|\cS|) \E[X]$.
\end{enumerate}
\end{lemma}

Using this lemma, we can prove \Thm{algbit}. 
We first give a sketch of the proof. Later we will formalize our claims. 
At the end of the stream, the output bit $b_m$ is $1$ if $|\cW_m| > 0$ and 
a wedge from $\cC_m$ is sampled.
Note that both $|\cW_m|$ and $|\cC_m|$ are random variables.

To deal with the first event, we apply \Lem{wedge-via-collisions} with $\cS$ being the set of all wedges.
So, $\EX[|\cW_m|] = {\singlestor \choose 2} (2W/m^2) \approx \singlestor^2W/m^2$. 
If $\singlestor \geq cm/\sqrt{W}$, then $\EX[|\cW_m|] \geq c$ (a large enough number).
Intuitively, the probability that $|\cW_m| = 0$ is very small, and this can be bounded using the 
concentration bound of \Lem{wedge-via-collisions}.

Now, suppose that $|\cW_m| > 0$.
The probability that $b_m = 1$ (which is $\EX[b_m]$) is exactly the fraction $|\cC_m|/|\cW_m|$.
Suppose we could approximate this by $\EX[|\cC_m|]/\EX[|\cW_m|]$. By \Clm{future}, $\cC_m$ is the set
of future-closed wedges, the number of which is $T$,
so \Lem{wedge-via-collisions} tells us that $\EX[|\cC_m|] = {\singlestor \choose 2} (2T/m^2)$.
Hence, $\EX[|\cC_m|]/\EX[|\cW_m|] = T/W = \gcc/3$.

In general, the value of $|\cC_m|/|\cW_m|$ might be different from $\EX[|\cC_m|]/\EX[|\cW_m|]$.
But $|\cC_m|$ and $|\cW_m|$ are reasonably concentrated (by the second part of \Lem{wedge-via-collisions}, so we can argue
that this difference is small.

\vspace{2ex}
{\bf Proof of \Thm{algbit}:}
As mentioned in the proof sketch, the output bit $b_m$ is $1$ if $|\cW_m| > 0$ and a wedge from $\cC_m$ is sampled.
For convenience, we will use $Y = |\cW_m|$ for  the total number of wedges formed by edges in $\cR_m$,  and  we will use $Z = |\cC_m|$ for  the number of future-closed wedges formed by edges in $\cR_m$. 
Both $Y$ and $Z$ are random variables
that depend on $\cR_m$. We apply \Lem{wedge-via-collisions} to understand the
behavior of $Y$ and $Z$. Let $\beta' = \beta/5$ ($\beta$ is the parameter
in the original \Thm{algbit}).

\begin{claim} \label{clm:exp-y} $\E[Y] = \singlestor(\singlestor-1)W/m^2$. With probability $>1-\beta'$,
$|Y - \EX[Y]| \leq \beta' \EX[Y]$.

Analogously, 
$\E[Z] = \singlestor (\singlestor-1)T/m^2$. With probability $>1-\beta'$, $|Z - \EX[Z]| \leq (\beta' W/T) \EX[Z]$.
\end{claim}

\begin{proof} First, we deal with $Y$. In \Lem{wedge-via-collisions}, let the set $\cS$ 
be the entire set of wedges. The random variable $X$ of the lemma is exactly $Y$,
and the size of $\cR_m$ is $\stor$.
So $\E[Y] = {\singlestor \choose 2} (2W/m^2) = \singlestor(\singlestor-1)W/m^2$. We set $\gamma$ in the second
part of \Lem{wedge-via-collisions} to be $\beta'$. By the premise of \Thm{algbit},
$\singlestor \geq cm/(\beta^3 \sqrt{T}) \geq  cm/(\beta^3 \sqrt{W})$. Moreover, for a large enough constant $c$, the latter is at least $c'm/(\beta'^3 \sqrt{W})$. 
We can apply the second part of \Lem{wedge-via-collisions} to derive the
weak concentration of $Y$.

For $Z$, we apply $\cS$ \Lem{wedge-via-collisions} with the set of future-closed wedges. These are exactly $T$
in number. An argument identical to the one above completes the proof.
\qed
\end{proof}

This suffices to prove the second part of \Thm{algbit}. We multiply the inequality $|Y - \E[Y]| \leq \beta' \E[Y]$
by $m^2/\singlestor(\singlestor-1)$, and note that the estimate is $est = m^2 |\cW|/(\singlestor(\singlestor-1))$.
Hence, $|est - W| \leq \beta'W$ with probability $>1-\beta'$.

We have proven that $\E[Z]/\E[Y] = T/W$ and would like to argue this is almost
true for $Z/Y$. This is formalized in the next claim.

\begin{claim} \label{clm:e} Suppose $\cE$ is the following event: $\max(|Y - \EX[Y]|,|Z-\EX[Z]|) \leq \beta' \EX[Y]$.
Then, $|\EX[Z/Y| \cE] - \gcc/3| \leq 4\beta'$. Furthermore, $\Pr[\cE] > 1-2\beta'$.
\end{claim}

\begin{proof} Since the deviation probabilities as given in \Clm{exp-y}
are at most $\beta'$, the union bound on probabilities implies $\Pr[\cE] >1 - 2\beta'$.

Since $\singlestor \geq cm/\beta^3 W$, by \Clm{exp-y}, $\EX[Y] \geq c'^2/\beta'^6$.
Hence, when $\cE$ happens, $Y > 0$. In other words, with probability at least $1-2\beta'$, the edges in $\cR_m$ will form a wedge.

Now look at $\EX[Z/Y| \cE]$.
When $\cE$ occurs, we can apply the bounds $|Y - \EX[Y]| \leq \EX[Y]$ and $|Z - \EX[Z]| \leq \EX[Z]$.
\begin{align*}
\frac{\EX[Z] - \beta' \EX[Y]}{(1+\beta')\EX[Y]} \leq \frac{Z}{Y} \leq \frac{\EX[Z] + \beta' \EX[Y]}{(1-\beta') \EX[Y]}
\end{align*}
We manipulate the upper bound with the following fact. For small enough $\beta'$, $1/(1-\beta') \leq 1+2\beta' \leq 2$.
Also, we use $\EX[Z]/\EX[Y] = T/W = \gcc/3$.
\begin{align*}
\frac{\EX[Z] + \beta' \EX[Y]}{(1-\beta') \EX[Y]} \leq (1+2\beta')\frac{\EX[Z]}{\EX[Y]} + 2\beta' = \gcc/3 + 4\beta'
\end{align*}
Using a similar calculation for the lower bound, when $\cE$ occurs,
$|Z/Y - \gcc/3| \leq 4\beta'$. Conditioned on $\cE$, $Z/Y \in [\gcc/3 - 4\beta', \gcc/3 + 4\beta']$,
implying $|\EX[Z/Y | \cE] - \gcc/3| \leq 4\beta'$. 
\end{proof}

We have a bound on $\EX[Z/Y | \cE]$, but we really care about $\EX[b_m]$.
The key is that conditioned on $Y > 0$, the expectation of $b_m$ is $Y/Z$,
and $Y > 0$ happens with large probability. We argue formally in \Clm{handle}
that $|\EX[b_m] - \EX[Z/Y | \cE] \leq \beta'$. Combined with \Clm{e},
we get $|\EX[b_m] - \gcc/3| \leq 5\beta' = \beta$, as desired.

\begin{claim}\label{clm:handle} $| \EX[b_m] - \EX[Z/Y | \cE] | \leq \beta'$.
\end{claim}
\begin{proof}
Let $\cF$ denote the event $Y > 0$. When $\cE$ holds, then $\cF$ also holds. Since {\algbit} outputs 0 when $\cF$ does not hold, we get $\EX[b_m | \overline{\cF} ] = 0$.   And since $\EX[b_m] = \EX[b_m | \cF] \Pr[\cF] + \EX[b_m | \overline{\cF} ] \Pr[\overline{\cF}]$, $\EX[b_m] =  \EX[b_m | \cF] \Pr[\cF] $. Further observe that $\EX[b_m | \cF]$ is exactly equal to $\EX[Z/Y | \cF]$.  Therefore, we get $\EX[b_m] = \EX[Z/Y | \cF] \Pr[\cF]$. By Bayes' rule,
\begin{align}
&\EX[b_m] = \EX[Z/Y | \cF] \Pr[\cF]  \notag\\
&= (\EX[Z/Y | \cF \cap \cE] \Pr[\cE | \cF] + \EX[Z/Y | \cF \cap \overline{\cE}] \Pr[\overline{\cE} | \cF] ) \cdot \Pr[\cF] \notag\\
&= \EX[Z/Y | \cF \cap \cE] \Pr[\cE \cap \cF] + \EX[Z/Y | \cF \cap \overline{\cE}] \Pr[\overline{\cE} \cap \cF] \notag\\
&= \EX[Z/Y | \cE] \Pr[\cE] + \EX[Z/Y | \cF \cap \overline{\cE}] \Pr[\overline{\cE} \cap \cF] \label{eqn:final}
\end{align}
The second last equality uses the fact that $\Pr[A \cap B] = \Pr[A | B] \cdot \Pr[B]$, while the last equality uses the fact that $\cF \cap \cE = \cE$ (since $\cE$ implies $\cF$). 
Note that $Z/Y \leq 1$. Thus, (\ref{eqn:final}) is at least $\EX[Z/Y | \cE] \Pr[\cE] \geq \EX[Z/Y | \cE]  (1 - \beta')
\geq \EX[Z/Y | \cE] - \beta'$. 
Moreover, (\ref{eqn:final}) is at most $\EX[Z/Y | \cE] + \Pr[\overline{\cE} \cap \cF]  \leq \EX[Z/Y | \cE]  + \Pr[\overline{\cE}] \leq \EX[Z/Y | \cE]  + \beta'$.
\end{proof}

\vspace{2ex}
{\bf Proof of \Lem{wedge-via-collisions}:}
The first part is an adaptation of the birthday paradox calculation.
Let the (multi)set $\cR = \{r_1, r_2, \ldots, r_{\stor}\}$.
We define random variables $X_{i,j}$ for each $i, j \in [\stor]$ with $i < j$.
Let $X_{i, j} = 1$ if the wedge $\set{r_i, r_j}$ belongs to $\cS$ and $0$ otherwise. Then $X = \sum_{i < j} X_{i,j}$. 

Since $\cR$ consists of uniform i.i.d. edges from $G$, the following holds: for every $i < j$ and every (unordered) pair of edges $\set{e_{\alpha}, e_{\beta}}$ from $E$, $\Pr[ \set{r_i, r_j}= \set{e_{\alpha}, e_{\beta}}] = 2/m^2$. This implies $\Pr[X_{i,j} = 1] = 2|\cS|/m^2$. By linearity of expectation and identical
distribution of all $X_{i,j}$s, $\E[X]$ $= {\singlestor \choose 2} \E[X_{1,2}]$ $= {\singlestor \choose 2} \Pr[X_{1,2} = 1]$ $= {\singlestor \choose 2} (2|\cS|/m^2)$, as required. 

The second part is obtained by applying the Chebyschev inequality. Let $Var[X]$ denote the variance of $X$.
For any $h > 0$,
\begin{align}
\displaystyle\Pr[|X - \E[X]| > h] \leq Var[X]/h^2 \label{eqn:chebyschev}
\end{align}

We need an upper bound on the variance of $X$ to apply \Eqn{chebyschev}. This is given
in \Lem{variance-bound}. Before proving the lemma, we use it to complete the main proof.
We set $h = (\gamma W/|\cS|)\EX[X]$. Note that $\EX[X] = \singlestor(\singlestor-1)|\cS|/m^2$,
so $h^2 = \gamma^2\stor^2(\stor-1)^2W^2/m^4 \geq \gamma^2\stor^4W^2/2m^4$.
By \Eqn{chebyschev}, $\Pr[|X - \E[X]| > h]$ is at most the following.
\begin{align*}
 \frac{Var[X]}{h^2} \leq \frac{18\stor^3 W^{3/2}/{m^3}}{\gamma^2\stor^4W^2/2m^4}
\leq \frac{36m/(\gamma^2 \sqrt{W})}{\stor} \leq \gamma
\end{align*}
where the final inequality holds $\singlestor \geq c'm/(\gamma^3\sqrt{W})$.

\begin{lemma}[Variance bound]\label{lem:variance-bound}
Assuming $W \geq m$ and $\singlestor \geq m/\sqrt{W}$, $$Var[X] \leq 18\stor^3 W^{3/2}/m^3.$$
\end{lemma}

\begin{proof} We use the same notation as in the proof of \Lem{wedge-via-collisions}. For convenience,
we set $\mu = \EX[X_{i,j}]$, which is $2|\cS|/m^2$. By the definition of variance and linearity of expectation, 
\begin{align*}
&Var[X]= \E[X^2] - (\E[X])^2 \\
&= \E[\sum_{i < j}X_{i,j} \sum_{p < q} X_{p,q}] - \mbox{${\singlestor \choose 2}^2$} \mu^2 
= \sum_{i < j, p < q}\E[X_{i,j} X_{p,q}] - \mbox{${\singlestor \choose 2}^2$} \mu^2.
\end{align*}
The summation is split as follows. 
\begin{align*}
\sum_{i< j, p < q}\E[X_{i,j} X_{p,q}] =
\sum_{i<j} \E[X^2_{i,j}] + \sum_{\substack{i < j, p < q \\ |\{i,j\} \cap \{p,q\}| = 1}} \E[X_{i,j} X_{p,q}] + \sum_{\substack{i < j, p < q \\ \{i,j\} \cap \{p,q\} = \emptyset}} \E[X_{i,j} X_{p,q}] \\
\end{align*}
We deal with each of these terms separately. For convenience, we refer to the terms (in order) as $A_1, A_2$, and $A_3$.
We first list the upper bounds for each of these terms and derive the final bound on $Var[X]$. 
\begin{asparaitem}
	\item $A_1 = {\singlestor \choose 2} \mu$.
	\item $A_2 \leq 12{\singlestor \choose 3} \sum_{v \in [n]} d^3_v/m^3$.
	\item $A_3 = 6{\singlestor \choose 4} \mu^2$.
\end{asparaitem}
We shall prove these shortly. From these, we directly bound $Var[X]$.
\begin{align*}
	Var[X] & = A_1 + A_2 + A_3 - {\singlestor \choose 2}^2 \mu^2 \\
	& \leq {\singlestor \choose 2} \mu + 12{\singlestor \choose 3} \sum_{v \in [n]} d^3_v/m^3 + 6{\singlestor \choose 4} \mu^2 - {\singlestor \choose 2}^2 \mu^2 
\end{align*}
Note that $6{\singlestor \choose 4} = \singlestor(\singlestor-1)(\stor-2)(\stor-3)/4 \leq [\singlestor(\singlestor-1)/2]^2$. Since the $\ell_3$-norm is less
than the $\ell_2$-norm, $\sum_v d^3_v \leq (\sum_v d^2_v)^{3/2}$. Since $W \geq m$,
we have $\sum_v d^2_v = 2 \sum_v {d_v \choose 2} + 2m \leq 4W$.
Plugging these bounds in (and using gross
upper bounds to ignore constants),
\begin{align*}
	Var[X] & \leq \stor^2 \mu + 2\stor^3 \frac{(\sum_v d^2_v)^{3/2}}{m^3}\\
	& \leq \frac{2 \stor^2 |\cS|}{m^2} + \frac{16 \stor^3 W^{3/2}}{m^3} \\
	& \leq \frac{2 \stor^2 W}{m^2} + \frac{16 \stor^3 W^{3/2}}{m^3} \leq \frac{18 \stor^3 W^{3/2}}{m^3}
\end{align*}
The final step uses the fact that $\singlestor \geq m/\sqrt{W}$.
\end{proof}

We now bound the terms $A_1$, $A_2$, and $A_3$ in three separate claims.

\begin{claim} \label{clm:A1} $\sum_{i < j} \EX[X^2_{i,j}] = {\singlestor \choose 2} \mu$.
\end{claim}

\begin{proof} Since $X_{i,j}$ only takes the values $0$ and $1$, $\EX[X^2_{i,j}] = \EX[X_{i,j}] = \mu$.
\end{proof}

\begin{claim} \label{clm:A2} $$\sum_{\substack{i < j, p < q \\ |\{i,j\} \cap \{p,q\}| = 1}} \E[X_{i,j} X_{p,q}] \leq 12{\singlestor \choose 3} \sum_{v \in [n]} d^3_v/m^3.$$
\end{claim}

\begin{proof} How many terms are in the summation? There are 3 distinct indices in $\{i,j,p,q\}$. For each distinct triple
of indices, there are 6 possible way of choosing $i < j, p < q$ among these indices such that $ |\{i,j\} \cap \{p,q\}| = 1$.
This gives $6{\singlestor \choose 3}$ terms. By symmetry, each term in the summation is equal to $\EX[X_{1,2} X_{1,3}]$.
This is exactly the probability that $\{r_1, r_2\} \in \cS$ and $\{r_1, r_3\} \in \cS$.
We bound this probability above by $2 \sum_{v \in [n]} d_v^3/m^3$, completing the proof.

Let $\cE$ be the event  $\set{r_1, r_2} \in S$ and  $\set{r_1, r_3} \in S$. Let $\cF$ be the event that edge $r_1$ intersects  edges $r_2$ and $r_3$. 
Observe that $\cE$ implies $\cF$. Therefore, it suffices to bound the probability of the latter event. We also use the inequality $\forall a,b, (a+b)^2 \leq 2(a^2 + b^2)$.
Also, note that the number of edges intersecting any edge $(u,v)$ is exactly $d_u + d_v - 1$.
\begin{align*}
&\Pr_{r_1,r_2, r_3}[\set{r_1, r_2} \in S, \set{r_1, r_3} \in S] \\
&\leq \Pr_{r_1,r_2, r_3}[ r_1 \cap r_2 \neq \emptyset \wedge r_1 \cap r_3 \neq \emptyset]\\
&= \sum_{(u,v) \in E}\Pr_{r_2, r_3}[\set{u,v} \cap r_2 \neq \emptyset \wedge \set{u,v} \cap r_3 \neq \emptyset | r_1 = (u,v)] \cdot \frac{1}{m}\\
&= \sum_{(u,v) \in E} \frac{ (d_u + d_v -1)^2}{m^2} \cdot \frac{1}{m}
\leq \frac{1}{m^3} \sum_{(u,v) \in E} { (d_u + d_v)^2} \\
&\leq \frac{2}{m^3} \sum_{(u,v) \in E} (d^2_u + d^2_v) = \frac{2 \sum_{v \in [n]} d_v^3}{m^3}
\end{align*}
For the final equality, consider the number of terms in the summation where $d^2_u$ appears. This is the number of edges $(u,v)$ (over all $v$), which
is exactly $d_u$.
\end{proof}

\begin{claim} \label{clm:A3} 
$$ \sum_{\substack{i < j, p < q \\ \{i,j\} \cap \{p,q\} = \emptyset}} \E[X_{i,j} X_{p,q}] = 6{\singlestor \choose 4} \mu^2 $$
\end{claim}

\begin{proof} There are $6{\singlestor \choose 4}$ terms in the summation (${\singlestor \choose 4}$ ways of choose $\{i,j,p,q\}$
and 6 different orderings). Note that $X_{i,j}$ and $X_{p,q}$ are independent, \emph{regardless} of the structure
of $G$ or the set of wedges $\cS$. This is because $\Pr[X_{i,j}=1 | X_{p,q} = 0] = \Pr[X_{i,j}=1 | X_{p,q} = 1] = \Pr[X_{i,j}=1]$.
In other words, the outcomes of the random edges $r_p$ and $r_q$ do not affect the edges $r_i, r_j$ (by independence
of these draws) and hence cannot affect the random variable $X_{i,j}$. Thus, $\E[X_{i,j} X_{p,q}] = \EX[X_{i,j}] \EX[X_{p,q}]$ $=\mu^2$.
\end{proof}

\subsection{Circumventing problems with \large{\algbit{}}} \label{sec:circum}

In this section we will discuss two problems that limit the practicality of \algbit{} and how \algstream{} circumvents these problems with heuristics.

\Thm{algbit} immediately gives a small sublinear space streaming algorithm for estimating $\gcc$. The output of
\algbit{} has almost the exact expectation. We can run many independent invocations of \algbit{} 
and take the fraction of $1$s to estimate $\EX[b_m]$ (which is close to $\gcc/3$).
A Chernoff bound tells us that $O(1/\epsilon^2)$ invocations suffice to estimate $\EX[b_m]$
within an additive error of $\epsilon$. The total space required by the algorithm becomes $O(m/(\sqrt{T}\epsilon^2))$, which can be very expensive in practice.
Even though $m/\sqrt{T}$ is not large, for the reasonable value of $\epsilon = 0.01$,
the storage cost blows up by a factor of $10^4$. This is the standard method used in previous work for streaming
triangle counts.

This blowup is avoided in \algstream{} by \emph{reusing} {\em the same reservoir of edges} for sampling wedges.
Note that \algbit{} is trying to generate a single uniform random wedge from $G$, and we use independent
reservoirs of edges to generate multiple samples. \Lem{wedge-via-collisions} says that for a reservoir
of $km/\sqrt{W}$ edges, we expect $k^2$ wedges. So, if $k > 1/\epsilon$ and we get $> 1/\epsilon^2$ wedges.
Since the reservoir contains a large set of wedges, we could just use a subset of these for estimating $\EX[b_m]$.
Unfortunately, these wedges are correlated with each other, and we cannot theoretically prove the desired concentration.
In practice,  the algorithm generates so many wedges that  downsampling these wedges for the wedge reservoir leads to a sufficiently uncorrelated sample, and  we get excellent results
by reusing the wedge reservoir. This is an important distinction all other streaming work~\cite{JoGh05,BuFrLeMaSo06,PavanVLDB2013}.
We can multiply our space by $1/\epsilon$ to (heuristically) get error $\epsilon$, but this is not possible
through previous algorithms. Their space is multiplied by $1/\epsilon^2$.

The second issue is that \algbit{} requires a fair bit of bookkeeping. We need to generate a random wedge from the large set $\cW_t$, the set of wedges formed by the current edge reservoir.
While this is possible by storing 
\edgeres{} as a subgraph, we have a nice (at least in the authors' opinion) heuristic fix that avoids these complications.

Suppose we have a uniform random wedge $w \in \cW_{t-1}$. We can convert it to an ``almost" uniform random wedge
in $\cW_t$. If $\cW_{t} = \cW_{t-1}$  (thus $\cN_t = \emptyset$ which is true most of the time), then $w$ is also uniform in $\cW_t$.
Suppose not. Note that $\cW_t$ is constructed by removing some wedges from $\cW_{t-1}$ and inserting $\cN_t$.
Since $w$ is uniform random in $\cW_{t-1}$, \emph{if $w$ is also present in $\cW_{t}$}, 
then it is uniform random in $\cW_{t} \setminus \cN_t$.
Replacing $w$ by a uniform random wedge in $\cN_t$ with probability $|\cN_t|/|\cW_t|$ yields
a uniform random wedge in $\cW_t$.
This is precisely what \algstream{} does.

When $w \notin \cW_t$, then the edge replaced by $e_t$ must be in $w$.
We approximate this as a low probability event and simply ignore this case. 
Hence, in \algstream{}, we simply assume that $w$ is always in $\cW_{t}$. This is technically
incorrect, but it appears to have little effect on the accuracy in practice.
And it leads to a cleaner, efficient implementation.

\section{Experimental Results} \label{sec:experiments}

We implemented our algorithm in C++ and ran our experiments on a MacBook Pro laptop equipped with a 2.8GHz Intel core i7 processor and 8GB memory. 

{\bf Predictions on various graphs:} We run \algstream{} on a variety of graphs obtained from the SNAP database \cite{Snap}. 
The vital statistics of all the graphs are provided in \Tab{gcc}. 
We simply set the edge reservoir size $\stor_e$ as 20K and wedge reservoir size $\stor_w$ as 20K for all our runs. Each graph is converted into
a stream by taking a random ordering of the edges.
In \Fig{cc-triangles}, we show
our results for estimating both the transitivity, $\gcc$ and  triangle count, $T$. The absolute values are plotted for $\gcc$ together with the true values.
For the triangle counts, we plot the relative error (so $|est - T|/T$, where $est$ is the algorithm output) for each graph, since
the true values can vary over orders of magnitude. 
Observe that the transitivity estimates are very accurate. The relative
error for $T$ is mostly below 8\%, and often below 4\%.

All the graphs listed have millions of edges, so our storage
is always 2 orders of magnitude smaller than the size of  graph. Most dramatically, we get accurate results on the {\tt Orkut} social network,
which has 220M edges. \emph{The algorithm stores only 40K edges, a 0.0001-fraction
of the graph.} Also observe the results on the {\tt Flickr} and {\tt Livejournal} graphs, which also run
into tens of millions of edges.
\begin{table}[t]	
 \caption{ Properties of the graphs used in the experiments}
\label{tab:gcc}
  \centering\small
  \begin{tabular}{|>{\tt}r@{\,}|*{6}{r@{\,}|}}
    \hline
            Graph &  $n$ &  $m$  &  $W$ & $T$ & $\gcc$ \\
    \hline
                          amazon0312 &   401K &  2350K &    69M &  3686K &	0.160  \\
              amazon0505 &   410K &  2439K &    73M &  3951K &  	0.162\\
             amazon0601 &   403K &  2443K &    72M &  3987K &  0.166\\
             as-skitter &  1696K & 11095K & 16022M & 28770K & 0.005	\\
                        cit-Patents &  3775K & 16519K &   336M &  7515K & 0.067\\
            	             roadNet-CA &  1965K &  2767K &     6M &   121K & 0.060	\\
          	           web-BerkStan &   685K &  6649K & 27983M & 64691K & 0.007\\
             web-Google &   876K &  4322K &   727M & 13392K & 0.055 \\
           web-Stanford &   282K &  1993K &  3944M & 11329K & 0.009\\
              wiki-Talk &  2394K &  4660K & 12594M &  9204K &  0.002 \\
              youtube & 1158K&  2990K& 1474M& 3057K& 0.006 \\
              flickr & 1861K& 15555K  &14670M & 548659K &   0.112 \\
              livejournal & 5284K & 48710K& 7519M& 310877K& 0.124 \\
              orkut & 3073K& 	223534K & 45625M	&627584K & 0.041 \\	
    \hline
 \end{tabular}
\end{table}

\begin{figure*}[tb]
  \centering
    \subfloat[{Transitivity}]{\begin{tikzpicture}\draw [fill=\typea,thin] (0.100,0) rectangle (0.220,1.926) ;
\draw [fill=\typeb,thin] (0.230,0) node[below]{\rotatebox[origin=t]{90}{\scriptsize amazon0312}} rectangle (0.350,1.951) ;
\draw [fill=\typea,thin] (0.460,0) rectangle (0.580,1.948) ;
\draw [fill=\typeb,thin] (0.590,0) node[below]{\rotatebox[origin=t]{90}{\scriptsize amazon0505}} rectangle (0.710,1.926) ;
\draw [fill=\typea,thin] (0.820,0) rectangle (0.940,1.987) ;
\draw [fill=\typeb,thin] (0.950,0) node[below]{\rotatebox[origin=t]{90}{\scriptsize amazon0601}} rectangle (1.070,1.991) ;
\draw [fill=\typea,thin] (1.180,0) rectangle (1.300,0.065) ;
\draw [fill=\typeb,thin] (1.310,0) node[below]{\rotatebox[origin=t]{90}{\scriptsize as-skitter}} rectangle (1.430,0.054) ;
\draw [fill=\typea,thin] (1.540,0) rectangle (1.660,0.806) ;
\draw [fill=\typeb,thin] (1.670,0) node[below]{\rotatebox[origin=t]{90}{\scriptsize cit-Patents}} rectangle (1.790,1.058) ;
\draw [fill=\typea,thin] (1.900,0) rectangle (2.020,0.725) ;
\draw [fill=\typeb,thin] (2.030,0) node[below]{\rotatebox[origin=t]{90}{\scriptsize roadNet-CA}} rectangle (2.150,0.824) ;
\draw [fill=\typea,thin] (2.260,0) rectangle (2.380,0.083) ;
\draw [fill=\typeb,thin] (2.390,0) node[below]{\rotatebox[origin=t]{90}{\scriptsize web-BerkStan}} rectangle (2.510,0.072) ;
\draw [fill=\typea,thin] (2.620,0) rectangle (2.740,0.663) ;
\draw [fill=\typeb,thin] (2.750,0) node[below]{\rotatebox[origin=t]{90}{\scriptsize web-Google}} rectangle (2.870,0.623) ;
\draw [fill=\typea,thin] (2.980,0) rectangle (3.100,0.103) ;
\draw [fill=\typeb,thin] (3.110,0) node[below]{\rotatebox[origin=t]{90}{\scriptsize web-Stanford}} rectangle (3.230,0.108) ;
\draw [fill=\typea,thin] (3.340,0) rectangle (3.460,0.026) ;
\draw [fill=\typeb,thin] (3.470,0) node[below]{\rotatebox[origin=t]{90}{\scriptsize wiki-Talk}} rectangle (3.590,0.025) ;
\draw [fill=\typea,thin] (3.700,0) rectangle (3.820,0.075) ;
\draw [fill=\typeb,thin] (3.830,0) node[below]{\rotatebox[origin=t]{90}{\scriptsize youtube}} rectangle (3.950,0.083) ;
\draw [fill=\typea,thin] (4.060,0) rectangle (4.180,1.346) ;
\draw [fill=\typeb,thin] (4.190,0) node[below]{\rotatebox[origin=t]{90}{\scriptsize flickr}} rectangle (4.310,1.444) ;
\draw [fill=\typea,thin] (4.420,0) rectangle (4.540,1.488) ;
\draw [fill=\typeb,thin] (4.550,0) node[below]{\rotatebox[origin=t]{90}{\scriptsize livejournal}} rectangle (4.670,1.606) ;
\draw [fill=\typea,thin] (4.780,0) rectangle (4.900,0.495) ;
\draw [fill=\typeb,thin] (4.910,0) node[below]{\rotatebox[origin=t]{90}{\scriptsize orkut}} rectangle (5.030,0.511) ;
\draw [<->, thick] (5.34,0) -- (0,0)-- (0,0.04) node[left] at (-0.6, 1.20) {\rotatebox{90}{\small Transitivity}} -- (0, 2.40) ;
\draw [dashed] (0, 0.480) node [left]{\scriptsize $0.04$} -- (5.360,0.480);
\draw [dashed] (0, 0.960) node [left]{\scriptsize$0.08$} -- (5.360,0.960);
\draw [dashed] (0, 1.440) node [left]{\scriptsize$0.12$} -- (5.360,1.440);
\draw [dashed] (0, 1.920) node [left]{\scriptsize$0.16$} -- (5.360,1.920);
\draw [fill=\typea,thin] (0.93,2.64) node[right] at(1.29,2.76) {\small Exact} rectangle (1.17,2.88)  ;
\draw [fill=\typeb,thin] (2.34,2.64) node[right]at(2.70,2.76) {\small Estimate} rectangle (2.58,2.88)  ;
\end{tikzpicture}}\hspace{4ex}
     \subfloat[{Triangles}]{\begin{tikzpicture}\draw [fill=\typea,thin] (0.100,0) node[below] at (0.200,-0.1) {\rotatebox[origin=t]{90}{\scriptsize amazon0312}} rectangle (0.300,0.400) ;
\draw [fill=\typea,thin] (0.400,0) node[below] at (0.500,-0.1) {\rotatebox[origin=t]{90}{\scriptsize amazon0505}} rectangle (0.600,1.200) ;
\draw [fill=\typea,thin] (0.700,0) node[below] at (0.800,-0.1) {\rotatebox[origin=t]{90}{\scriptsize amazon0601}} rectangle (0.900,1.000) ;
\draw [fill=\typea,thin] (1.000,0) node[below] at (1.100,-0.1) {\rotatebox[origin=t]{90}{\scriptsize as-skitter}} rectangle (1.200,1.400) ;
\draw [fill=\typea,thin] (1.300,0) node[below] at (1.400,-0.1) {\rotatebox[origin=t]{90}{\scriptsize cit-Patents}} rectangle (1.500,0.800) ;
\draw [fill=\typea,thin] (1.600,0) node[below] at (1.700,-0.1) {\rotatebox[origin=t]{90}{\scriptsize roadNet-CA}} rectangle (1.800,1.400) ;
\draw [fill=\typea,thin] (1.900,0) node[below] at (2.000,-0.1) {\rotatebox[origin=t]{90}{\scriptsize web-BerkStan}} rectangle (2.100,2.400) ;
\draw [fill=\typea,thin] (2.200,0) node[below] at (2.300,-0.1) {\rotatebox[origin=t]{90}{\scriptsize web-Google}} rectangle (2.400,0.800) ;
\draw [fill=\typea,thin] (2.500,0) node[below] at (2.600,-0.1) {\rotatebox[origin=t]{90}{\scriptsize web-Stanford}} rectangle (2.700,1.400) ;
\draw [fill=\typea,thin] (2.800,0) node[below] at (2.900,-0.1) {\rotatebox[origin=t]{90}{\scriptsize wiki-Talk}} rectangle (3.000,0.400) ;
\draw [fill=\typea,thin] (3.100,0) node[below] at (3.200,-0.1) {\rotatebox[origin=t]{90}{\scriptsize youtube}} rectangle (3.300,0.600) ;
\draw [fill=\typea,thin] (3.400,0) node[below] at (3.500,-0.1) {\rotatebox[origin=t]{90}{\scriptsize flickr}} rectangle (3.600,1.200) ;
\draw [fill=\typea,thin] (3.700,0) node[below] at (3.800,-0.1) {\rotatebox[origin=t]{90}{\scriptsize livejournal}} rectangle (3.900,1.000) ;
\draw [fill=\typea,thin] (4.000,0) node[below] at (4.100,-0.1) {\rotatebox[origin=t]{90}{\scriptsize orkut}} rectangle (4.200,1.000) ;
\draw [<->, thick] (4.50,0) -- (0,0)-- (0,0.06) node[left] at (-0.6, 1.60) {\rotatebox{90}{\small Relative error}} -- (0, 3.00) ;
\draw [dashed] (0, 0.600) node [left]{\scriptsize $0.03$} -- (4.500,0.600);
\draw [dashed] (0, 1.200) node [left]{\scriptsize$0.06$} -- (4.500,1.200);
\draw [dashed] (0, 1.800) node [left]{\scriptsize$0.09$} -- (4.500,1.800);
\draw [dashed] (0, 2.400) node [left]{\scriptsize$0.12$} -- (4.500,2.400);
\end{tikzpicture}}
   \caption{Output of a single run of \algstream{} on a variety of real datasets with 20K edge reservoir and 20K wedge reservoir. The plot on the left gives the estimated transitivity values (labelled streaming) alongside their exact values. The plot on the right gives the relative error of \algstream's estimate on triangles $T$. Observe that the relative error for $T$ is mostly below $8\%$, and often below $4\%$.}
  \label{fig:cc-triangles}
\end{figure*}
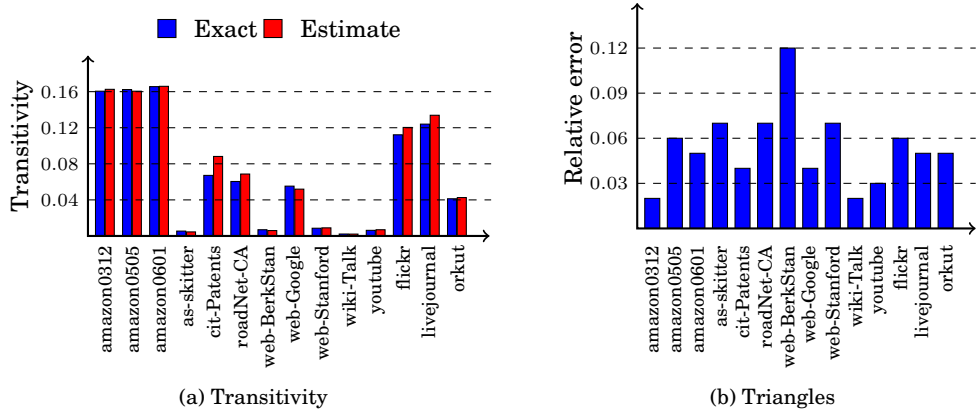

{\bf Real-time tracking:} A benefit of \algstream{} is that it can maintain a real-time estimate of $\gcc_t$ and $T_t$.
We take a real-world temporal graph, {\tt cit-Patents}, which contains patent citation data over a 40 year period.  The vertices of this graph are the patents and the edges  correspond to the citations. The edges
are time stamped with the year of citation and hence give a stream of edges. Using an edge reservoir
of 50K and wedge reservoir of 50K, we accurately track these values over time (refer to \Fig{tracking}). 
Note that this is still orders of magnitude smaller than the full size of the graph, which is 16M edges. 
The figure only shows the true values and the estimates  for the year ends.
As the figure shows the estimates are consistently accurate over time.
\medskip

{\bf Convergence of our estimate:} 
We demonstrate that our algorithm converges to the true value as we increase the space. We run our algorithm on amazon0505 graph by 
increasing the space ($\stor_e + \stor_w$) available to the algorithm. For convenience, we keep the size of edge reservoir and wedge reservoir the same. 
In \Fig{se-equals-sw-concentration}, estimates for transitivity and triangles rapidly converge to the true value. Accuracy increases   with more storage for unto 10,000 edges, but after that stabilizes.
We get similar results for other graphs, but do not provide all details for brevity.

\begin{figure*}
  \centering
  \subfloat[Convergence of the transitivity estimate]{\includegraphics[scale=0.35]{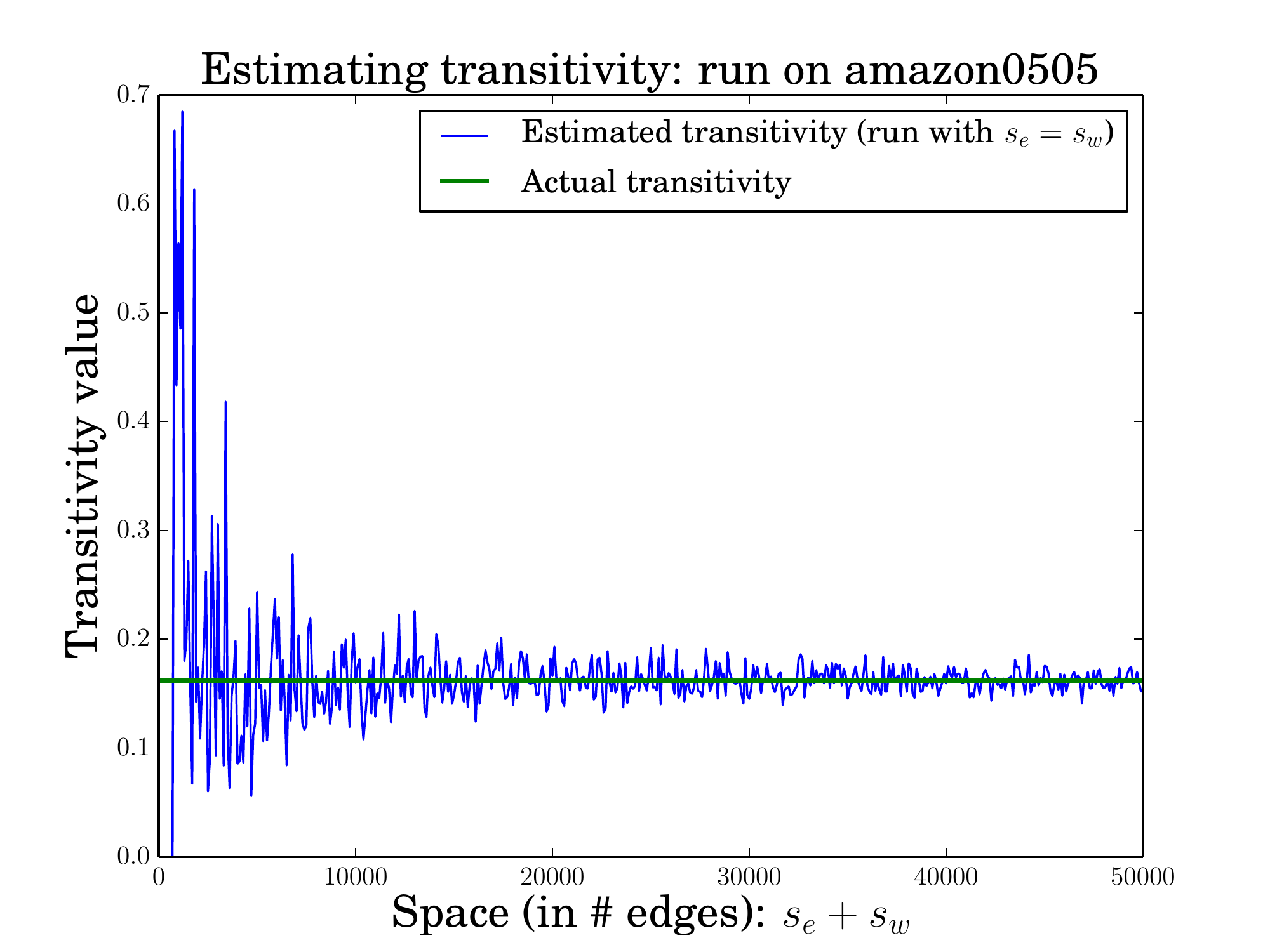}}
  \subfloat[Convergence of triangles estimate]{\includegraphics[scale=0.35]{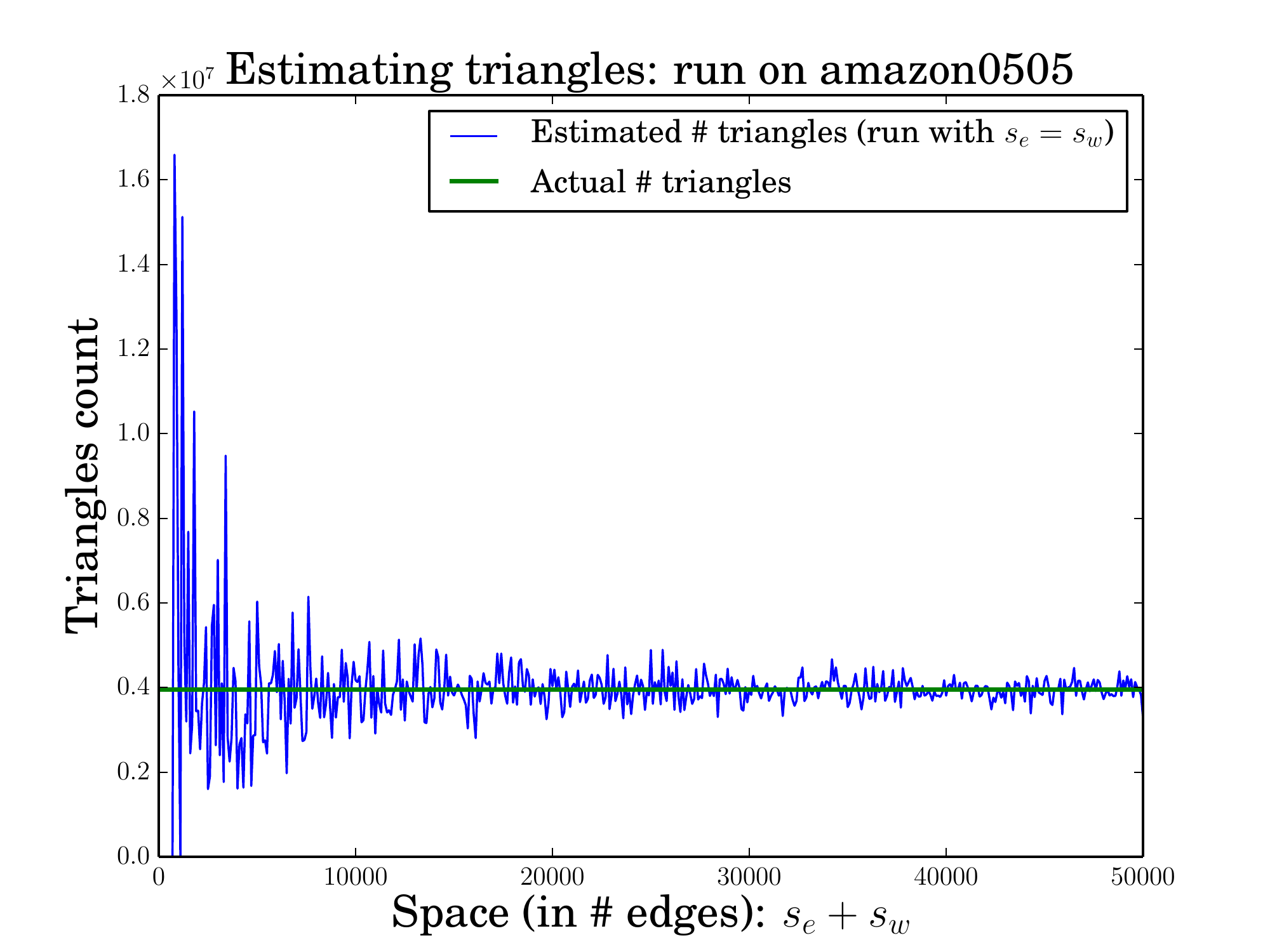}}
   \caption{Concentration of estimate on amazon0505: We run our algorithm keeping the size of edge reservoir and wedge reservoir the same. We plot the transitivity and triangles estimate and observe that they converge to the true value.}
  \label{fig:se-equals-sw-concentration}
\end{figure*}

{\bf Effects of storage on estimates:} We explore the effect that the  sizes of the edge reservoir, $\stor_e$ and the wedge reservoir, $\stor_w$ have on 
the quality of the estimates for $\gcc$. In the first experiment we fix $\stor_e$ to 10K and 20K and increase $\stor_w$.  The results are presented 
in \Fig{varyinga}.
 In this figure, for any point $x$ on the horizontal axis, the corresponding point on the vertical axis is the average error in $[1,x]$.  
In all cases, the error decreases as we increase $\stor_w$. However, it decreases
sharply initially but then flattens showing that the marginal benefit of increasing $\stor_w$ beyond improvements diminish, and it does not help to only increase $\stor_w$. 

In \Fig{varyingb},
we fix $\stor_w$ to 10K and 20K and increase $\stor_e$.  The results are similar to the first case. 

\begin{figure*}[tb]
  \centering
  \subfloat[ \label{fig:varyinga}Varying reservoir wedges]{\includegraphics[scale=0.35]{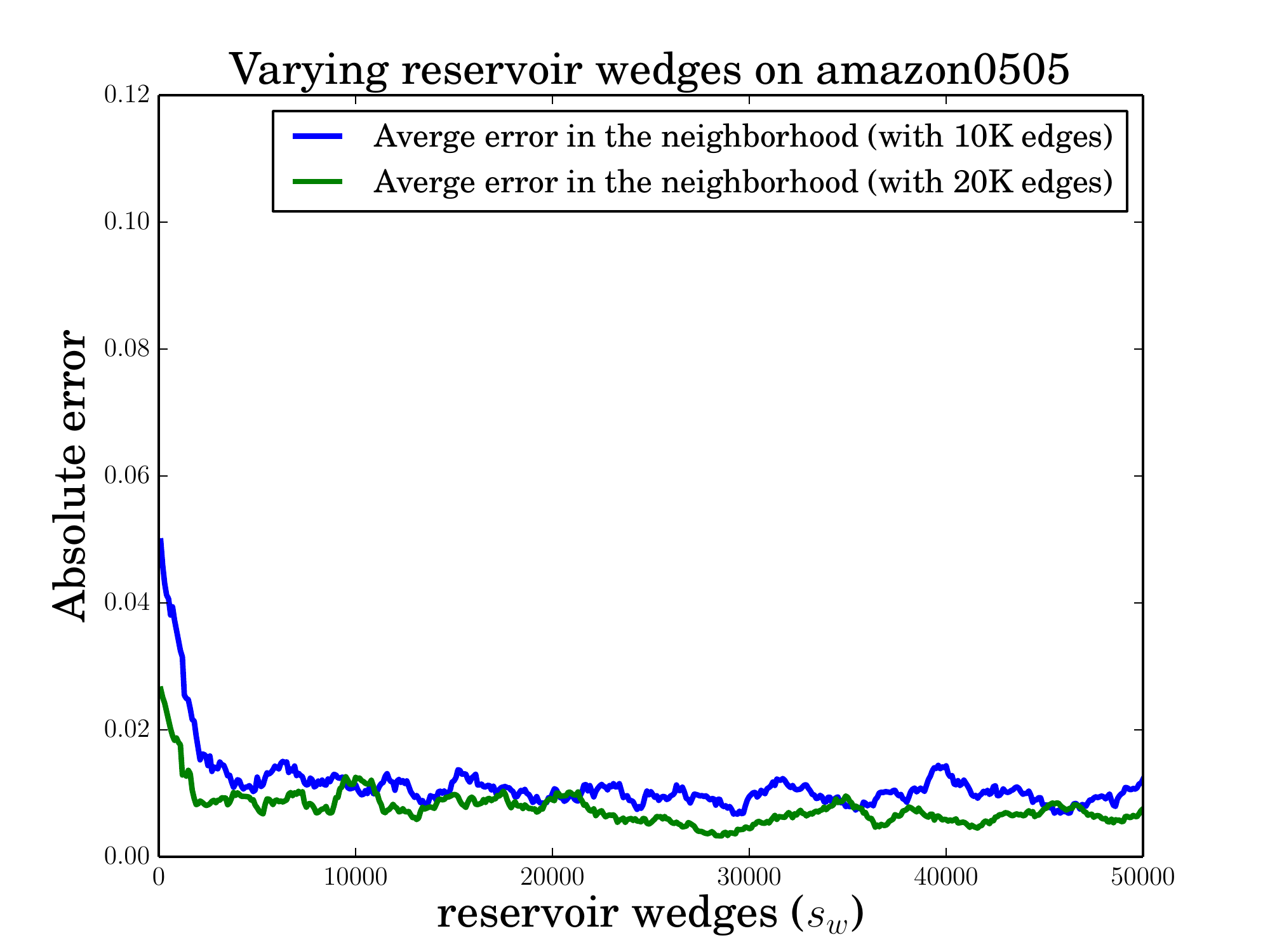}}
  \subfloat[\label{fig:varyingb}Varying reservoir edges]{\includegraphics[scale=0.35]{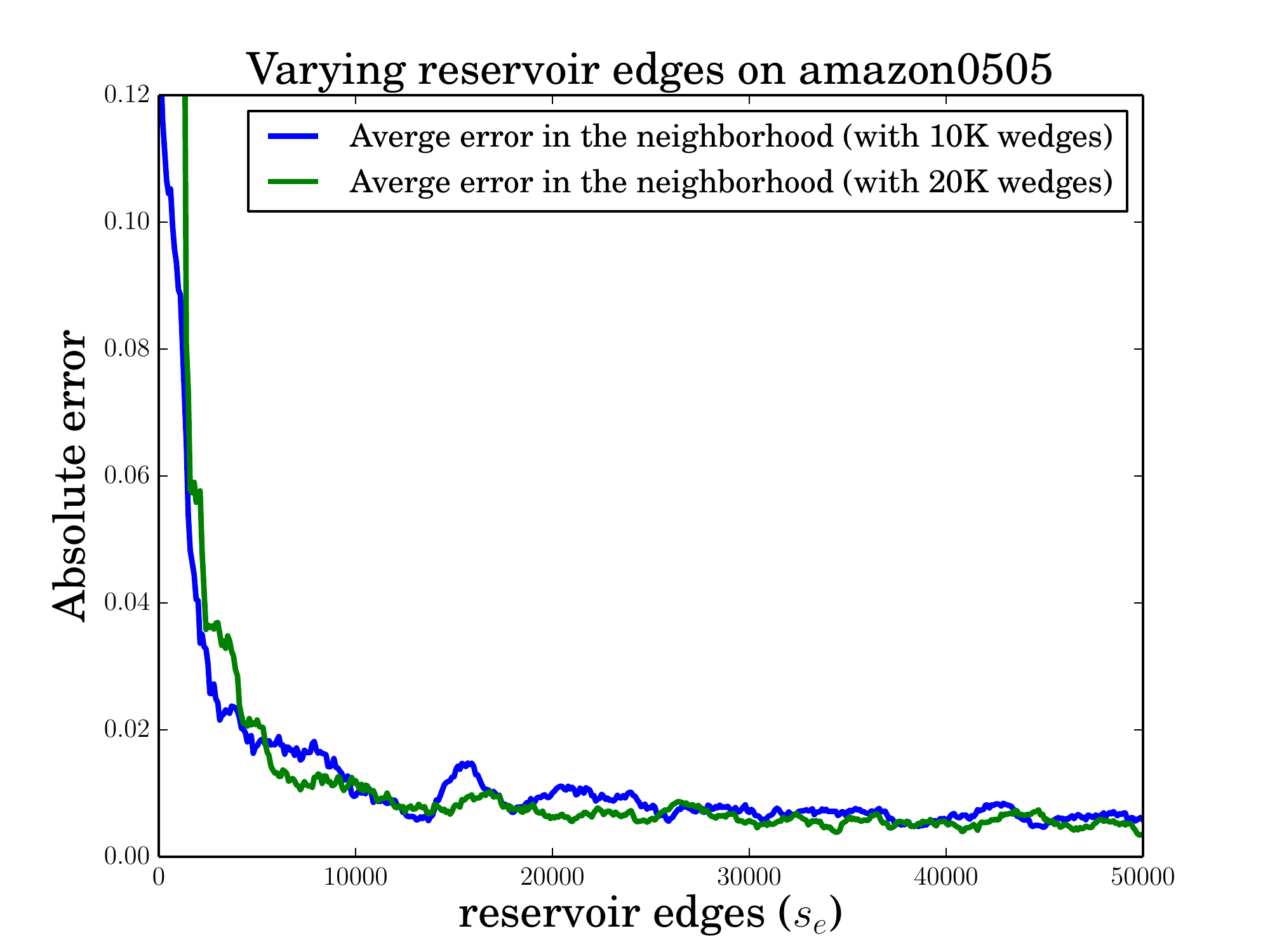}}
  \caption{ How accuracy of transitivity is affected by  varying one of $s_e$ and $s_w$ while keeping the other fixed on amazon0505.}
  \label{fig:varying}
\end{figure*}

{\bf Effect of stream ordering:} 
In this set of experiments, we investigate  the effect of the stream order on the accuracy of the estimates. 
For this purpose, we generate a set of different orderings of the edges in amazon0505 and run \algstream{} on these orderings.
The results are given in \Tab{various-orderings}. We fix the edge and wedge reservoir to 20K and use the following orderings.
\begin{table}[h]
\caption{Run of our algorithm on various orderings of the same graph (web-NotreDame). Each run is made with parameters $s_e = s_w = 20K$. }\label{tab:various-orderings}
\centering
\begin{tabular}{| l |c|c|}
\hline
Orderings									& Absolute error in transitivity	&	Relative error in triangles	\\
\hline
Random permutation							&	0.00035	&	2.39\%							\\
BFS and rest							 	&	0.00775 	&	5.81\%							\\
DFS and rest								&	0.0004   	&	0.88\%							\\
Degree sorted							 	&	0.0007     	& 	0.30\%							\\
Reverse degree sorted						&	0.00385	&	4.56\%							\\
\hline
\end{tabular}
\end{table}
The first ordering is a random ordering. Next, we generate a stream through a breadth first search (bfs) as follows.
We take a bfs tree from a random vertex and list out all edges in the tree. Then, we list the remaining edges in random order.
Our third ordering involves taking a depth first search (dfs) from a random vertex and list out edges in order as seen by the dfs. Finally, 
the next two orderings are obtained by sorting the vertices by degree (increasing and decreasing respectively) and listing all edges incident to a vertex. Note that the last two orderings are incidence streams.

\algstream{} performs well on all these different orderings. There is little deviation in the transitivity values. 
There is somewhat more difference in the triangle numbers, but it never exceeds 5\% relative error.
Overall the results show  that the accuracy often algorithm is invariant to the stream order. 

{\bf The performance of \algbit{}:} Does our heuristic in \algstream{} really help over
independent invocations of \algbit{}? We hope to have convinced the reader of the rapid convergence
of \algstream{}. We implement \algbit{} by simply setting $\stor_e$ to 5K
and $\stor_w$ to $1$. We then run multiple copies of it and output three times the fraction of $1$s (which is an unbiased
estimate for $\gcc$). As shown in \Fig{parallel-runs}, convergence is poor. This is because even when
using a space of 250,000 edges, we only have 250,000/5,000 = 50 independent wedge samples, which
is too small to get a close estimate.

\begin{figure*}[tb]
  \centering
  \includegraphics[scale=0.5]{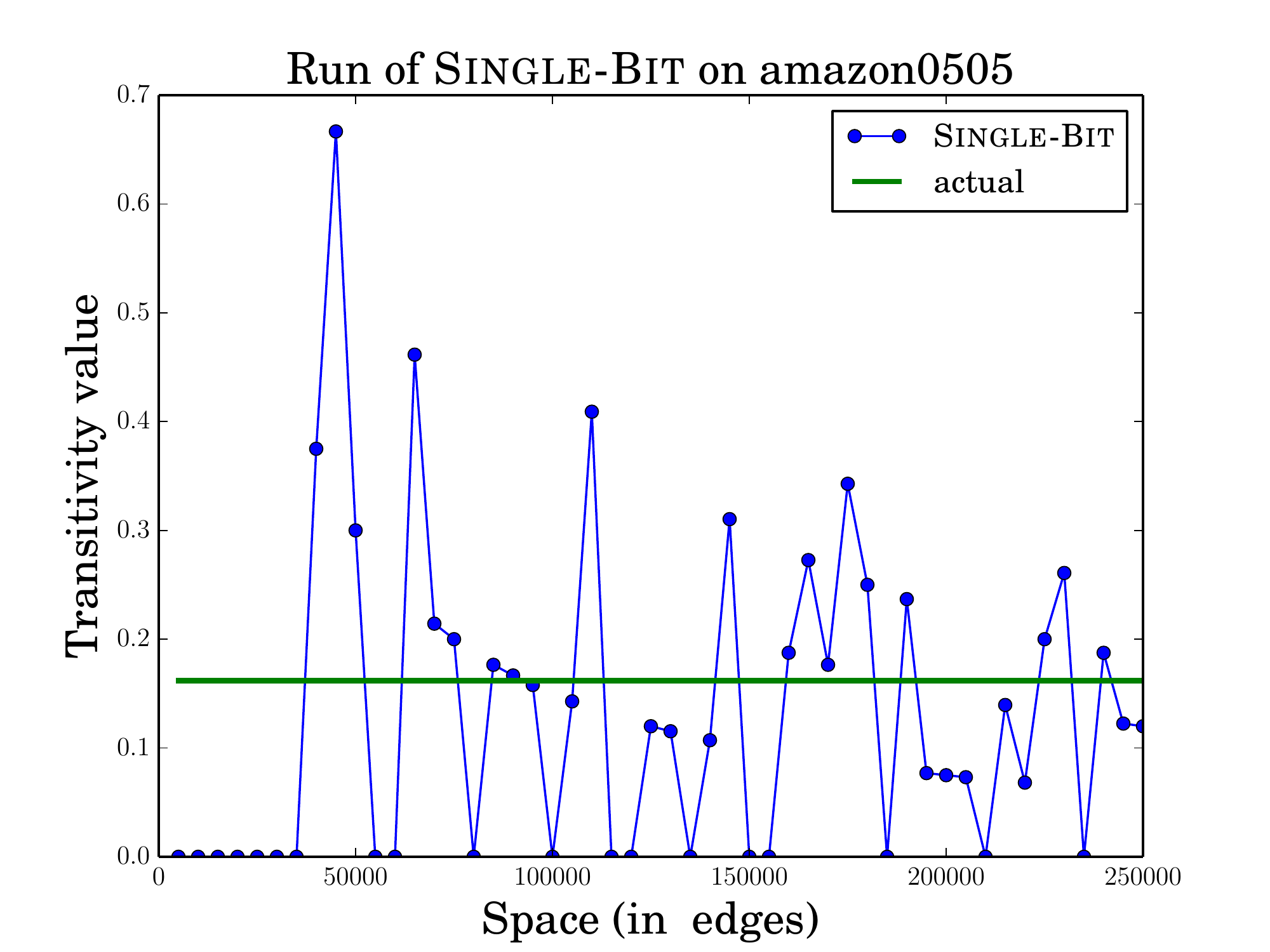}
   \caption{Independent runs on \algbit{} on amazon0505: We fix the edge reservoir $s_e$ to 5K and the wedge reservoir $s_w$ to 1. We plot transitivity estimate obtained by taking the average over independent runs.}
  \label{fig:parallel-runs}
\end{figure*}

{\bf Comparison with previous work:} The streaming algorithm of Buriol et al.~\cite{BuFrLeMaSo06} was implemented and run on real graphs.
The basic sampling procedure involves sampling a random edge and a random vertex and trying the complete a triangle. This
is repeatedly independently in parallel to get an estimate for the number of triangles. Buriol et al. provide various heuristics to speed
up their algorithm, but the core sampling procedure is what was described above.
In general, their algorithms gets fairly large error even with storage of 100K edges. For the amazon0505 graph, and it took a storage of 70K
edges just to get a non-zero answer. Even after 100K edges, the triangles counts had more than 100\% error. (This is consistent with
the experimental results given in~\cite{BuFrLeMaSo06}.)

\section{Conclusion}
\label{sec:conc}
Our streaming algorithm is practical and gives accurate answers, but only works for simple graphs.  
A natural future direction is to consider the streaming setting when the input is a directed graph and/or a multigraph. 

We expect to generalize these ideas to maintain
richer information about triangles. For example, could we maintain degree-wise
clustering coefficients in a single pass? It is likely that these ideas can be used to
counting different types of triangles in directed or other attributed graphs. 

At a higher level, the sampling approach may be useful for other properties. 
We can infer the transitivity of a massive graph by maintaining a small subsample of the edges.
It remains to be seen what other properties can be inferred by similar sampling schemes.

\begin{acks}
The first author would like to thank Shubham Gupta for many helpful discussions.
\end{acks}

\bibliographystyle{ACM-Reference-Format-Journals}
\bibliography{streaming_triangles}  
\end{document}